%% file: main.tex
\documentclass[a4paper,UKenglish]{llncs}

\usepackage{microtype}

\usepackage{xspace}
\usepackage{amsfonts,amsmath,amssymb}
\usepackage{graphicx}

\usepackage{algorithm}
\usepackage[noend]{algorithmic}

 \usepackage[defblank]{paralist}
\usepackage{subfig}
\usepackage{mfirstuc}

\newcommand{\remove}[1]{}

\newcommand{\NPhard}{NP-hard\xspace}

\newcommand{\threesat}{\textsc{3-sat}\xspace}

\newcommand{\GR}{GR\xspace}
\newcommand{\GRexpall}{GR-EA\xspace}

\newcommand{\undirected}{Undirected\xspace}

\newcommand{\SBGP}{S-BGP\xspace}
\newcommand{\SoBGP}{So-BGP\xspace}
\newcommand{\ASSET}{AS-SET\xspace}
\newcommand{\simpleHijack}{origin-spoofing\xspace}
\newcommand{\assetBgp}{\ASSET BGP\xspace}
\newcommand{\plainBgp}{plain BGP\xspace}
\newcommand{\hijack}{\textsc{hijack}\xspace}

\newcommand{\interception}{\textsc{interception}\xspace}

\newcommand{\intLengthStruct}{\textsc{Intermediate} structure\xspace}
\newcommand{\shoLengthStruct}{\textsc{Short} structure\xspace}
\newcommand{\longLengthStruct}{\textsc{Long} structure\xspace}
\newcommand{\disruptiveLengthStruct}{\textsc{Disruptive} structure\xspace}

\newcommand{\xxmakefirstuc}[1]{\expandafter\xmakefirstuc\expandafter{#1}}
\newcommand{\xxxmakefirstuc}[1]{\expandafter\xxmakefirstuc\expandafter{#1}}

\title{Computational Complexity of Traffic Hijacking under BGP and S-BGP\thanks{The original publication is available at \url{www.springerlink.com}.}}

\author{Marco Chiesa\inst{1} \and Giuseppe {Di Battista}\inst{1} \and \\ Thomas Erlebach\inst{2} \and Maurizio Patrignani\inst{1}}
\institute{Dept. of Computer Science and Automation, Roma Tre University
 \email{\{chiesa,gdb,patrigna\}@dia.uniroma3.it},
\and Dept. of Computer Science, University of Leicester\\
\email{t.erlebach@mcs.le.ac.uk}
}

\authorrunning{M. Chiesa et al.} 

\begin{document}

\maketitle

\begin{abstract}
Harmful Internet hijacking incidents put in evidence how fragile the Border
Gateway Protocol (BGP) is, which is used to exchange routing information between
Autonomous Systems (ASes). As proved by recent research contributions, even
\SBGP, the secure variant of BGP that is being deployed, is not fully able to
blunt traffic attraction attacks. Given a traffic flow between two ASes, we
study how difficult it is for a malicious AS to devise a strategy for hijacking
or intercepting that flow. We show that this problem marks a sharp difference
between BGP and \SBGP. Namely, while it is solvable, under reasonable
assumptions, in polynomial time for the type of attacks that are usually
performed in BGP, it is NP-hard for \SBGP. 
Our study has several by-products. E.g., we solve a problem left open
in the literature, stating when performing a hijacking in \SBGP is equivalent to
performing an interception.
\end{abstract}

\input{intro}

\input{bgp_models}


\input{origin_spoofing}
\input{sbgp}




\section{Conclusions and Open Problems}\label{sect:conclusions}

Given a communication flow between two ASes we studied how difficult it is for a malicious AS $m$ to devise a strategy for hijacking or intercepting that flow. This problem marks a sharp difference between BGP and \SBGP. Namely, while in a realistic scenario the problem is computationally tractable for typical BGP attacks it is NP-hard for \SBGP. This gives new evidence of the effectiveness of the adoption of \SBGP. It is easy to see that all the NP-hardness results that we obtained for the hijacking problem easily extend to the interception problem. Further, we solved a problem left open in~\cite{gshr-hsasirp-10}, showing when performing a hijacking in \SBGP is equivalent to performing an interception.




Several problems remain open:
\begin{inparaenum}
\item We focused on a unique $m$. How difficult is it to find a strategy involving several malicious ASes~\cite{gshr-hsasirp-10}?
\item  In~\cite{szr-pbrpcnhr-10} it has been proposed to disregard the AS-paths length in the BGP decision process. How difficult is it to find an attack strategy in this different model?
\end{inparaenum}







\bibliographystyle{splncs03}
\bibliography{bibliography}


\end{document}

%% file: intro.tex
\section{Introduction and Overview}\label{sect:intro}


On 24th Feb.\ 2008, Pakistan Telecom started an unauthorized announcement of prefix 208.65.153.0/24~\cite{u-phy-08}. This announcement was propagated to the rest of the Internet, which resulted in the \emph{hijacking} of YouTube traffic on a global scale. Incidents like this put in evidence how fragile is the \emph{Border Gateway Protocol (BGP)}~\cite{rfc4271}, which is used to exchange routing information between Internet Service Providers (ISPs). Indeed, performing a hijacking attack is a relatively simple task. It suffices to issue a BGP announcement of a victim prefix from a border router of a malicious (or unaware) \emph{Autonomous System (AS)}. Part of the traffic addressed to the prefix will be routed towards the malicious AS rather than to the intended destination. 
A mischievous variation of the hijacking is the \emph{interception} when, after passing through the malicious AS, the traffic is forwarded to the correct destination. This allows the rogue AS to eavesdrop or even modify the transit packets.


In order to cope with this security vulnerability, a variant of BGP, called \SBGP~\cite{kls-sbgp-00}, has been proposed, that requires a PKI infrastructure both to validate the correctness of the AS that originates a prefix and to allow an AS to sign its announcements to other ASes. In this setting an AS cannot forge announcements that do not derive from announcements received from its neighbors.
However, \cite{gshr-hsasirp-10} contains surprising results: (i) simple hijacking strategies are 
tremendously effective and (ii) finding a strategy that maximizes the amount of traffic that is hijacked is NP-hard both for BGP and for \SBGP. 

%


In this paper we tackle the hijacking and interception problems from a new
perspective. Namely, given a traffic flow between two ASes, how difficult is it
for a malicious AS to devise a strategy for hijacking or intercepting at least
that specific flow? We show that this problem marks a sharp difference between
BGP and \SBGP. Namely, while it is polynomial time solvable, under reasonable
assumptions, for typical BGP attacks, it is \NPhard for \SBGP. This gives new
complexity related evidence of the effectiveness of the adoption of \SBGP. 
Also, we solve an open problem~\cite{gshr-hsasirp-10}, showing when
every hijack in \SBGP results in an interception.
%
Tab.~\ref{tab:recap1} summarizes our results. Rows correspond to different settings for a malicious AS $m$. 
The \simpleHijack setting (Sect.~\ref{sect:origin-spoofing}) corresponds to a
scenario where $m$ issues BGP announcements pretending to be the owner of a
prefix. Its degree of freedom is to choose a subset of its neighbors for such a
bogus announcement. This is the most common type of hijacking attack to
BGP~\cite{wikipedia}.
In \SBGP (Sect.~\ref{sect:sbgp}) $m$ must enforce the constraints imposed by
\SBGP, which does not allow to pretend to be the owner of a prefix that is
assigned to another AS. 
%
\begin{table}[t]
   \centering
   \begin{tabular}{|c|c|c|c|}
      \cline{2-4}
      \multicolumn{1}{l|}{} & { AS-paths of} & { Bounded AS-path} & { Bounded AS-path} \\
      \multicolumn{1}{l|}{} & { any length} & { length} & { length and AS degree} \\
      \hline
      {\bf  \xxxmakefirstuc{\simpleHijack}} &  \NPhard\scriptsize(Thm.~\ref{theo:grexpall_simple-hijack_bgp})&  P\scriptsize(Thm.~\ref{theo:grexpall_plb_origin_bgp})&  P \\
      \hline
      {\bf  \SBGP} &  \NPhard &  \NPhard\scriptsize(Thm.~\ref{theo:grexpall_plb_sbgp}) &  P\scriptsize(Thm.~\ref{theo:grexpall_plb_ndb_sbgp}) \\
      \hline
\end{tabular}
\vspace{1mm} 
   \caption{Complexity of finding a \hijack strategy in different settings.}
   \label{tab:recap1}
\end{table}
Columns of Tab.~\ref{tab:recap1} correspond to different assumptions about the Internet. In the first column we assume that the longest \emph{valley-free} path (i.e.\ a path enforcing certain customer-provider constraints) in the Internet can be of arbitrary length. 
This column has a theoretical interest since the length of the longest path
(and hence valley-free path) observed in the Internet remained constant even
though the Internet has been growing in terms of active AS numbers during the
last 15 years~\cite{h-brtar-11}.
Moreover, in today's Internet about $95 \%$ of the ASes is reached in $3$ AS
hops~\cite{h-brtar-11}. 
%
Hence, the second column corresponds to a quite realistic Internet, where the AS-path length is bounded by a constant. 
In the third column we assume that the number of neighbors of $m$ is bounded by a constant. This is typical in the periphery of the Internet.
%
A ``P'' means that a Polynomial-time algorithm exists. Since moving from left to right the setting is more constrained, we prove only the rightmost NP-hardness results, since they imply the NP-hardness results to their left. Analogously, we prove only the leftmost ``P'' results.

%% file: bgp_models.tex
\subsection{A Model for BGP Routing}\label{sect:bgp-models}

As in previous work on interdomain hijacking~\cite{gshr-hsasirp-10}, we model the Internet as a graph $G=(V,E)$. A vertex in $V$ is an \emph{Autonomous System (AS)}. Edges in $E$ are \emph{peerings} (i.e., connections) between ASes. A vertex owns one or more \emph{prefixes}, i.e., sets of contiguous IP numbers. The routes used to reach prefixes are spread and selected via BGP.
Since each prefix is handled independently by BGP, we focus on a single prefix $\pi$, owned by a destination vertex~$d$.


BGP allows each AS to autonomously specify which paths are forbidden (\emph{import policy}), how to choose the best path among those available to reach a destination (\emph{selection policy}), and a  subset of neighbors to whom the best path should be announced (\emph{export policy}).
%
%
%
BGP works as follows. Vertex $d$ initializes the routing process by sending \emph{announcements} to (a subset of) its neighbors. Such announcements contain $\pi$ and the \emph{path} of $G$ that should be traversed by the traffic to reach $d$. In the announcements sent from $d$ such a path contains just $d$. 
We say that a path $P=(v_n\ \dots\ v_0)$ is \emph{available} at vertex $v$ if 
$v_n$ announces $P$ to $v$.
Each vertex checks among its available paths that are not filtered by the import policy, which is the best one according to its selection policy, and then it announces that path to a set of its neighbors in accordance with the export policy.
Further, BGP has a loop detection mechanism, i.e., each vertex $v$ ignores a route if $v$ is already contained in the route itself. 


Policies are typically specified according to two types of relationships~\cite{h-ipas-99}. In a \emph{customer-provider} relationship, an AS that wants to access the Internet pays an AS which sells this service. In a \emph{peer-peer} relationship two ASes exchange traffic 
without any money transfer between them. Such commercial relationships between ASes are represented by orienting a subset of the edges of $E$. Namely, edge $(u,v) \in E$ is directed from $u$ to $v$ if $u$ is a customer of $v$, while it is undirected if $u$ and $v$ are peers. A path is \emph{valley-free} if provider-customer and peer-peer edges are only followed by provider-customer edges.


The Gao-Rexford~\cite{gr-sirgc-00} Export-all (\emph{GR-EA}) conditions are commonly assumed to hold in this setting~\cite{gshr-hsasirp-10}.
\begin{inparablank} 
 \item {\bf GR1}: $G$ has no directed cycles that would correspond to unclear customer-provider roles.
 \item {\bf GR2}: Each vertex $v \in V$ sends an announcement containing a path $P$ to a neighbor
$n$ only if path $(n\ v)P$ is valley-free. 
Otherwise, some AS would provide transit to either its peers or its providers without revenues.
\item {\bf GR3}: A vertex prefers paths through customers over those provided by peers and paths through peers over those provided by providers.
\item {\bf Shortest Paths}: Among paths received from neighbors of the same class (customers, peers, and provider), a vertex chooses the shortest ones.
\item {\bf Tie Break}: 
If there are multiple such paths, a vertex chooses according to some tie break rule. As in~\cite{gshr-hsasirp-10}, we assume that the one whose next hop has lowest AS number is chosen. Also, as in~\cite{es-wairg-11}, to tie break equal class and equal length simple paths $P_1^u=(u\ v)P_1^v$ and $P_2^u=(u\ v)P_2^v$ at the same vertex $u$ from the same neighbor $v$, if $v$ prefers $P_1^v$ over $P_2^v$, then $u$ prefers $P_1^u$ over $P_2^u$. This choice is called \emph{policy consistent} in~\cite{es-wairg-11} and it can be proven that it has the nice property  of making the entire set of policies considered in this paper policy consistent.
%
%
\item {\bf NE policy}: a vertex always exports a path except when GR2 forbids it to do so.
\end{inparablank}

Since we assume that the \GRexpall conditions are satisfied, then a (partially directed) graph is sufficient to fully specify the policies of the ASes. Hence, in the following a \emph{BGP instance} is just a graph.

\subsection{Understanding Hacking Strategies}\label{sect:bgp-models}

We consider the following problem. A BGP instance with three specific vertices, $d$, $s$, and $m$ are given, where such vertices are: the AS originating a prefix $\pi$, a source of traffic for $\pi$, and an attacker, respectively.
All vertices, but $m$, behave correctly, i.e., according to the BGP protocol and \GRexpall conditions. Vertex $m$ is interested in two types of attacks: \emph{hijacking} and \emph{interception}. In the hijacking attack $m$'s goal is to attract to itself at least the traffic from $s$ to $d$.  In the interception attack $m$'s goal is to be traversed by at least the traffic from $s$ to $d$. 




In Fig.~\ref{fig:example_of_attack} $(2,6)$ is peer-to-peer and the other edges are customer-provider. Prefix $\pi$ is owned and announced by $d$. According to BGP, the traffic from $s$ to $d$ follows $(s\ 6\ 2\ 1\ d)$. In fact, $2$ selects  $(1\ d)$. Vertex $6$ receives a unique announcement from $d$ (it cannot receive an announcement with $(5\ 4\ 3\ m\ 2\ 1\ d)$ since it is not valley-free).
By cheating, ({\bf Example 1}) $m$ can deviate the traffic from $s$ to $d$ attracting traffic from $s$. In fact, if $m$ pretends to be the owner of $\pi$ and announces it to $2$, then $2$ prefers, for shortest-path, $(2\ m)$ over $(2\ 1\ d)$. Hence, the traffic from $s$ to $d$ is received by $m$ following $(s\ 6\ 2\ m)$. A hijack!



Observe that $m$ could be smarter ({\bf Example 2}). Violating GR2, it can announce $(2\ 1\ d)$ to $3$. Since each of $3$, $4$ and $5$ prefers paths announced by customers (GR3), the propagation of this path is guaranteed. Therefore, $6$ has two available paths, namely, $(2\ 1\ d)$ and $(5\ 4\ 3\ m\ 2\ 1\ d)$. The second one is preferred because $5$ is a customer of $6$, while $2$ is a peer of $6$. Hence, the traffic from $s$ to $d$ is received by $m$ following path $(s\ 6\ 5\ 4\ 3\ m)$. Since after passing through $m$ the traffic reaches $d$ following $(m\ 2\ 1\ d)$ this is an interception.


\begin{figure}[t]
\begin{minipage}[t]{0.49\textwidth}%
\includegraphics[width=1\textwidth]{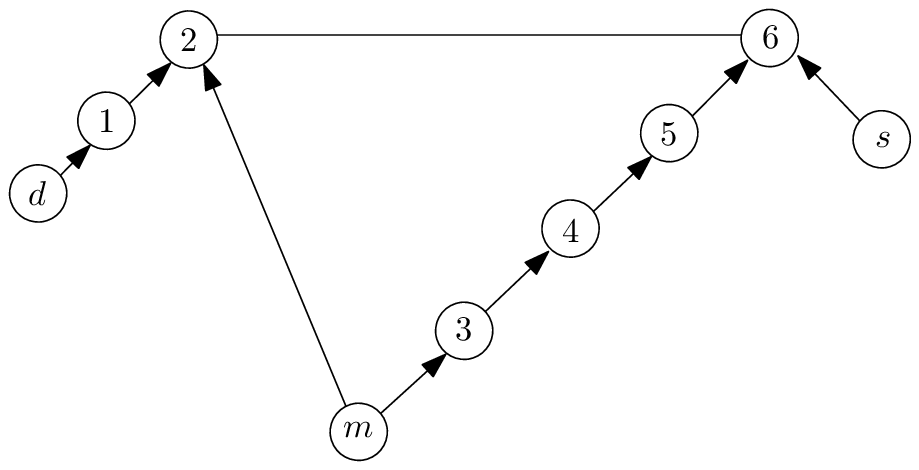}
\caption{A network for Examples 1 and 2.}\label{fig:example_of_attack}
\includegraphics[width=1\textwidth]{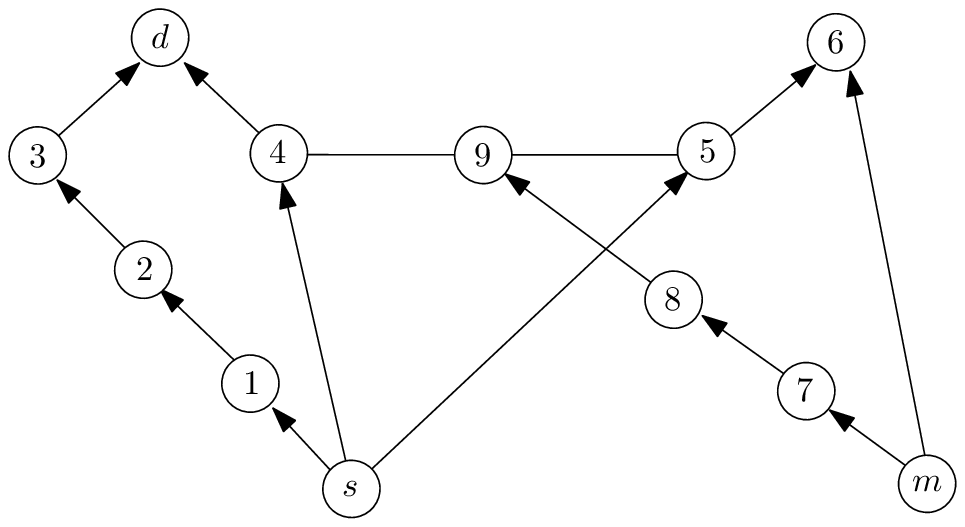}
\caption{A network for Example 3.}\label{fig:example_of_attack_2}
\end{minipage}%
\hfill
\begin{minipage}[t]{0.49\textwidth}%
\vspace{-30mm} 
\includegraphics[width=1\textwidth]{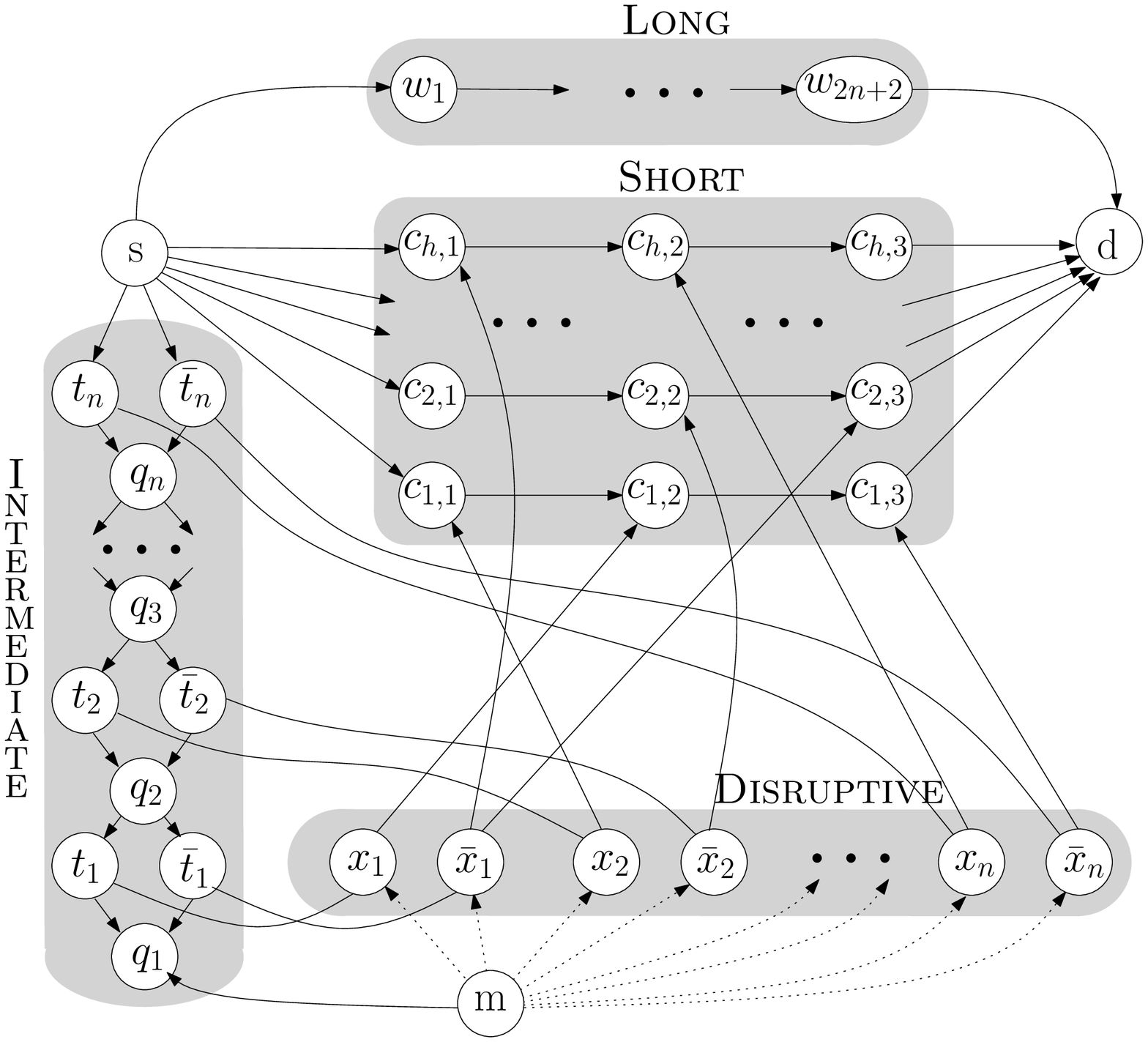}
\caption{Reduction of the \threesat problem to the \hijack problem when $m$ has \simpleHijack capabilities. Dotted lines from $m$ to vertices $x_i$ and $\bar x_i$ have length $2n+2$.}\label{fig:bgp_origin_spoofing_bgp_grexpall}
\end{minipage}
\end{figure}


Fig.~\ref{fig:example_of_attack_2} allows to show a negative example ({\bf Example 3}). According to BGP, the traffic from $s$ to $d$ follows $(s\ 4\ d)$. In fact, $s$ receives only paths $(4\ d)$ and $(1\ 2\ 3\ d)$, both from a provider, and prefers the shortest one. 
Suppose that $m$ wants to hijack and starts just announcing $\pi$ to $6$. Since all the neighbors of $s$ are providers, $s$ prefers, for shortest path, $(4\ d)$ over $(5\ 6\ m)$ (over $(1\ 2\ 3\ d)$ over $(4\ 9\ 8\ 7\ m)$) and the hijack fails. 
But $m$ can use another strategy.
Since $(s\ 5\ 6\ m)$ is shorter than $(s\ 1\ 2\ 3\ d)$, $m$ can attract traffic if $(4\ d) $ is ``disrupted'' and becomes not available at $s$. 
This happens if $4$ selects, instead of $(d)$, a path received from its peer neighbor $9$ ($m$ may announce that it is the originator of $\pi$ also to $7$). 
However, observe that if $4$ selects path $(4\ 9\ 8\ 7\ m)$ then $5$ selects path $(5\ 9\ 8\ 7\ m)$ since it is received from a peer and stops the propagation of $(s\ 5\ 6\ m)$.  Hence, $s$ still selects path $(s\ 1\ 2\ 3\ d)$ and the hijack fails. 

In order to cope with the lack of any security mechanism in BGP, several variations of the protocol have been proposed by the Internet community. One of the most famous, \SBGP, uses both origin authentication and cryptographically-signed announcements in order to guarantee that an AS announces a path only if it has received this path in the past.


The attacker $m$ has more or less constrained \emph{cheating capabilities}. 
\begin{inparaenum}
\item With the \emph{\simpleHijack} cheating capabilities $m$ can do the typical BGP announcement manipulation. I.e., $m$ can  pretend to be the origin of  prefix $\pi$ owned by $d$, announcing this to a subset of its neighbors. 
%
%
\item With the \emph{\SBGP} cheating capabilities $m$ must comply with the \SBGP constraints. I.e.: (a) $m$ cannot pretend to be the origin of  prefix $\pi$; and (b) $m$ can announce a path $(m\ u)P$ only if $u$ announced $P$ to $m$ in the past. However, $m$ can still announce paths that are not the best to reach $d$ and can decide to announce different paths to different neighbors. In Example 2, $m$ has \SBGP cheating capabilities.
%
%
\end{inparaenum}


In this paper we study the computational complexity of the \hijack and of the
\interception problems. The \hijack problem is formally defined as follows. {\bf
Instance}: A BGP instance $G$, a source vertex $s$, a destination vertex $d$, a
manipulator vertex $m$, and a cheating capability for $m$. {\bf Question}: Does
there exist a set of announcements that $m$ can simultaneously send to its
neighbors, according to its cheating capability, that produces a stable state
for $G$ where the traffic from $s$ to $d$ goes to $m$?
The \interception problem is defined in the same way but changing ``the traffic from $s$ to $d$ goes to $m$'' to ``the traffic from $s$ to $d$ passes through $m$ before reaching $d$''.



\subsection{Notation and Definitions}
%
%
We introduce some technical notation in order to prove our lemmas and theorems. A \emph{ranking function} determines the level of preference of paths available at vertex $v$. If $P_1, P_2$ are available at $v$ and $P_1$ is \emph{preferred} over $P_2$ we write $P_1<_\lambda^v P_2$.
The \emph{concatenation} of two nonempty paths $P=(v_k\ v_{k-1}\ \dots\ v_i)$, $k \geq i$, and $Q=(v_i\ v_{i-1}\ \dots\ v_0)$, $i \geq 0$, denoted as $PQ$, is the path $(v_k\ v_{k-1}\ \dots\ v_{i+1}\ v_i\ v_{i-1}\ \dots\ v_0)$.
Also, let $P$ be a valley-free path from vertex $v$. We say that $P$ \emph{is of
class} $3$, $2$, or $1$ if its first edge connects $v$ with a customer, a peer,
or a provider of $v$, respectively.
We also define a function $f^v$ for each vertex $v$, that maps each path from
$v$ to the integer of its class. Given two paths $P$ and $P'$ available at $v$
if $f^v(P)>f^v(P')$ we say that the class of $P$ \emph{is better than} the class
of $P'$.
In stable routing state $S$, a path $P=(v_1\ \dots\ v_n)$ is \emph{disrupted at vertex} $v_i$ by a path $P'$ if there exists a vertex $v_i$ of $P$ such that $v_i$ selects path $P'$.
%
Also, if $P'$ is preferred over $(v_i\ \dots\ v_n)$ because of the GR3 condition, we say that path $P$ \emph{is disrupted by a path of a better class}. 
Otherwise, if $P'$ is preferred over $(v_i\ \dots\ v_n)$ because of the shortest-paths criterion, we say that  $P$ is disrupted by a path of the \emph{same class}. 
%


\subsection{Routing Stability under Manipulator Attacks}
BGP policies can be so complex that there exist configurations that do not allow to reach any stable routing state (see, e.g., \cite{gsw-sppir-02}). A routing state is \emph{stable} if there exists a time $t$ such that after $t$ no AS changes its selected path. If the \GRexpall conditions are satisfied~\cite{gr-sirgc-00}, then a BGP network always converges to a stable state. However, there is a subtle issue to consider in attacks. As we have seen in the examples, $m$ can deliberately ignore the \GRexpall conditions. 
Anyway, the following lemma makes it possible, in our setting, to study the \hijack and the \interception problem ignoring stability related issues. First, we introduce some notation.

\newcommand{\lemmaStableState}{
Let $G$ be a BGP instance and suppose that at a certain time a manipulator $m$ starts announcing steadily any set of arbitrary paths to its neighbors. Routing in $G$ converges to a stable state.}

\begin{lemma}\label{lemma:grexpall_plain_stable_state}
\lemmaStableState
\end{lemma}

\begin{proof}
Suppose, for a contradiction, that, after $m$ starts its announcements, routing in $G$ is not stable. 
Let $u_0,\dots,u_n$ be a circular sequence of vertices such that: (1) each $u_i$ does not steadily announce a path; (2) the most preferred path $P^{u_i} = R_i Q_i = (r_1^{u_i}\ \dots\ r_{k_i}^{u_i})Q_{i} $ at $u_i$ that is available infinitely many times is such that each vertex in $Q_i$ but $r_{k_i}^{u_i}$ steadily announces a path; and (3) $r_{k_i}^{u_i}=u_{i+1}$, where $i$ has to be interpreted modulo $n$. Such a circular sequence is called dispute-wheel and it has been proved in~\cite{gsw-pdpvp-99} that if a system (with no manipulators) is not stable then it contains a dispute-wheel. We prove that the presence of a dispute-wheel in a GR-EA instance leads to a constradiction. Hence, a GR-EA instance always converges to a stable state. Observe that $|R_i| \geq 2$, otherwise $u_i$ would be stable. Since $m$ and $d$ steadily announce some paths and because $G$ is finite, such a sequence exists.
Observe that for each $i$, we have that $Q_{i-1}>^{u_i}_\lambda P^{u_{i}}$.

Suppose that for each $u_i$ we have that $P^{u_{i}}$ is preferred over $Q_{i-1}$ either by shortest path or by tie-break, i.e., $f^{u_{i}}(Q_{i-1}) = f^{u_{i}}(P^{u_{i}})$ and 
$|Q_{i-1}|\ge|P^{u_{i}}|$. Inequality $|R_{i}| \geq 2$ implies $|P^{u_{i}}|>|Q_{i}|$. Hence, we have $|Q_{i-1}|\ge|P^{u_{i}}|>|Q_{i}|\ge|P^{u_{i+1}}|$. Following the cycle of inequalities we have a contradiction as we obtain $|Q_{i}|>|Q_{i}|$. 

Conversely, let $u_m$ be a vertex that prefers $P^{u_{m}}$ over $Q_{m-1}$ for better class, that is $f^{u_{m}}(P^{u_{m}}) > f^{u_{m}}(Q_{m-1})$. For each $i=0,\dots,m-1,m+1,\dots,n$ we have $f^{u_{i}}(P^{u_{i}}) \ge f^{u_{i}}(Q_{i-1}) \ge f^{u_{i-1}}(P^{u_{i-1}})$ because of the GR2 and GR3 conditions. Following the cycle of inequalities we have a contradiction as we obtain $f^{u_{m}}(P^{u_{m}}) >f^{u_{m}}(P^{u_{m}})$.
\qed
\end{proof}

The existence of a stable state (pure Nash equilibrium) in a game where one player can deviate from a standard behavior has been proved, in a different setting in~\cite{es-wairg-11}. Such a result and Lemma~\ref{lemma:grexpall_plain_stable_state} are somehow complementary since the export policies they consider are more general than Export-All, while the convergence to the stable state is not guaranteed (even if such a stable state is always reachable from any initial state).

%% file: origin_spoofing.tex
\section{Checking if an Origin-Spoofing BGP Attack Exists}\label{sect:origin-spoofing}


In this section, we show that, in general, it is hard to find an attack strategy if $m$ has an \simpleHijack cheating capability (Theorem~\ref{theo:grexpall_simple-hijack_bgp}), while the problem turns to be easier in a realistic setting (Theorem~\ref{theo:grexpall_plb_origin_bgp}). 


\remove{

The exponential complexity of this approach is due to two
factors: (i) even in the case when a single path is announced, $m$ has
to decide to what neighbors the announcement is sent, which results in
an exponential number of possibilities w.r.t.\ the number of neighbors;
(ii) in the case $m$ has \plainBgp capabilities, a single announcement
can be forged in an exponential number of ways w.r.t.\ the number of
vertices of the graph. Hence, the following result in the case the Internet graph has no bound constraints may be not surprising at all.
}

A hijacking can be obviously found in exponential time by a simple brute force approach which simulates every possible attack strategy and verifies its effectiveness. The following result in the case the Internet graph has no bound constraints may be somehow expected.

\newcommand{\greBgpSimpleHijack}{
If the manipulator has \simpleHijack cheating capabilities, then problem \hijack is \NPhard.}

\begin{theorem}\label{theo:grexpall_simple-hijack_bgp}
\greBgpSimpleHijack
\end{theorem}

\begin{proof}
We prove that \hijack is \NPhard by a reduction from the \threesat problem. Let $F$ be a logical formula in conjunctive normal form with variables $X_1 \dots X_n$ and clauses $C_1 \dots C_h$ where each clause $C_i$ contains three literals. We construct a \GRexpall compliant BGP instance $G$ as follows. 

Graph $G$ consists of $4$ structures: the \intLengthStruct, the \shoLengthStruct, the \longLengthStruct, and the \disruptiveLengthStruct. See Fig.~\ref{fig:bgp_origin_spoofing_bgp_grexpall}.

The \intLengthStruct is the only portion of $G$ containing  valley-free paths joining $s$ and $m$ that are shorter than the one contained in the \longLengthStruct. It is composed by edge $(m,q_1)$ and two directed paths from $s$ to $q_1$ of length $2n$: the first path is composed by edges $(s,t_n)$, $(t_n,q_n)$, $(q_n,t_{n-1})$, $(t_{n-1},q_{n-1})$, $(q_{n-1},t_{n-2})$, \dots, $(t_2,q_2)$, $(q_2,t_1)$, and $(t_1,q_1)$ while the second path is composed by edges $(s,\bar t_n)$, $(\bar t_n,q_n)$, $(q_n,\bar t_{n-1})$, $(\bar t_{n-1},q_{n-1})$, $(q_{n-1},\bar t_{n-2})$, \dots, $(\bar t_2,q_2)$, $(q_2,\bar t_1)$, and $(\bar t_1,q_1)$. Obviously, these two paths can be used to construct an exponential number of other paths. We say that a path traverses the \intLengthStruct if it passes through vertices $s$ and $q_1$.

The \shoLengthStruct consists of $h$ paths joining $s$ and $d$. Each path has length $4$ and has edges $(s,c_{i,1})$, $(c_{i,1},c_{i,2})$, $(c_{i,2},c_{i,3})$, and $(c_{i,3},d)$ ($1\leq i\leq h$).
The \longLengthStruct is a directed path of length $2n+3$ with edges $(s,w_1)$, $(w_1,w_2)$, \dots, $(w_{2n+1},w_{2n+2})$, and~$(w_{2n+2},d)$. 
The \disruptiveLengthStruct is composed by $2n$ paths plus $3h$ edges. The $2n$ paths are defined as follows. For $1\leq i\leq n$ we define two paths. The first path contains a directed subpath of length $2n+2$ from $m$ to $x_i$ (dotted lines in Fig.~\ref{fig:bgp_origin_spoofing_bgp_grexpall}), plus the undirected edge $(x_i,t_i)$. The second path contains a directed subpath of length $2n+2$ from $m$ to $ \bar x_i$ (dotted lines in Fig.~\ref{fig:bgp_origin_spoofing_bgp_grexpall}) plus the undirected edge $(\bar x_i,\bar t_i)$. The $3h$ edges are added to $G$ as follows. For each clause $C_i$ and each literal $L_{i,j}$ of $C_i$, which is associated to a variable $x_k$, if $L_{i,j}$ is positive, then we add $(x_k,c_{i,j})$, otherwise we add $(\bar x_k,c_{i,j})$.
We say that a path traverses the \disruptiveLengthStruct if it traverses it from $m$ to $s$.


Vertices $s$, $d$, and $m$ have source, destination, and manipulator roles, respectively. 

Intuitively, the proof works as follows. The paths that allow traffic to go from $s$ to $m$ are only those passing through the \disruptiveLengthStruct and the \intLengthStruct. Also, the paths through the \intLengthStruct are shorter than the one through the \longLengthStruct, which is shorter than those through the \disruptiveLengthStruct.

If $m$ does not behave maliciously, $s$ receives only paths that traverse the \shoLengthStruct and the \longLengthStruct. In this case $s$ selects one of the paths in the \shoLengthStruct according to its tie break policy. 

Observe that if $m$ wants to attract traffic from $s$, then:
\begin{inparaenum}[(i)]
 \item a path from $m$ traversing entirely the \intLengthStruct has to reach $s$ and
 \item all paths contained in the \shoLengthStruct have to be disrupted by a path announced by~$m$.
\end{inparaenum}

Observe that only valley-free paths contained in the \intLengthStruct, which have length at least $2n+2$, can be used to attract traffic from $s$.  If (i) does not hold, then $s$ selects the path contained in the \longLengthStruct or a path contained in the \shoLengthStruct.
If (ii) does not hold, then $s$ selects a path contained in the \shoLengthStruct.

Our construction is such that the \threesat formula is satisfiable iff $m$ can attract the traffic from $s$ to $d$. To understand the interplay between our construction and the \threesat problem, consider (see Fig.~\ref{fig:bgp_origin_spoofing_bgp_grexpall}) the behavior of $m$ with respect to neighbors $x_2$ and $\bar x_2$. If $m$ wants to disrupt path $(s\ c_{1,1}\ c_{1,2}\ c_{1,3}\ d)$ (which corresponds to making clause $C_1$ true) it might announce the prefix to $x_2$. This would have the effect of disrupting $(s\ c_{1,1}\ c_{1,2}\ c_{1,3}\ d)$ by better class.  
Observe that at the same time this would disrupt all the paths through $t_2$. If $m$ is able to disrupt all the paths in the \shoLengthStruct, then $s$ has to select a path in the \intLengthStruct. However, $m$ has to be careful for two reasons. First, $m$ has to announce the prefix to $q_1$ (otherwise no path can traverse the \intLengthStruct). Second, $m$ cannot announce the prefix both to $x_2$ and to $\bar x_2$ (variable $X_2$ cannot be true and false at the same time). In this case, all the paths through $t_2$ and $\bar t_2$ are disrupted. Also, consider that the paths that reach $s$ through $t_2$ and $x_2$ ($\bar t_2$ and $\bar x_2$) and that remain available are longer than the one in the \longLengthStruct.

Now we show that if $F$ is satisfiable, then $m$ can attract traffic from $s$.
Let $M$ be a truth assignment to variables $X_1$, \dots, $X_n$ satisfying
formula $F$. Let $m$ announce to its neighbors paths as follows: if $X_i$ ($i=1,\dots,n$) is true then $m$ announces the prefix to $x_i$ and does not announce anything to $\bar x_i$; otherwise $m$ does the opposite. Also, the prefix in announced to $q_1$ in all cases.

We have that:
\begin{inparaenum}
\item all paths (one for each clause) in the \shoLengthStruct are disrupted by better class from the paths in the \disruptiveLengthStruct;
\item one path belonging to the \intLengthStruct is available at $s$;
\item the path in the \longLengthStruct, available at $s$, is longer than the path in the \intLengthStruct.
\end{inparaenum}
Hence, $m$ can attract traffic from $s$.

Now we prove that if manipulator $m$ can attract traffic from $s$, then $F$ is satisfiable.

We already know from the above discussion that $m$ can attract traffic from $s$ only using paths that traverse the \intLengthStruct entirely. We also know that these paths are longer than paths contained in the \shoLengthStruct and therefore, every path contained in the \shoLengthStruct has to be disrupted.
We have that paths contained in the \shoLengthStruct can be disrupted only by using paths contained in the \disruptiveLengthStruct. Let $V^*$ be the set of neighbors of $m$ different from $q_1$ that receive an announcement of the prefix from $m$. Observe that $s$, to attract traffic from $m$, has to announce the prefix to $q_1$. From the above discussion we have that for $i=1,\dots,n$ it is not possible both for $x_i$ and for $\bar x_i$ to receive the announcement. Also, since all paths in the \shoLengthStruct have been disrupted, for $j=1,\dots,h$ at least one of the $c_{j,k}$ ($k=1,2,3$) receives an announcement of the prefix from $m$. Hence, we define an assignment $M$, which satisfies formula $F$, as follows: for each $i=1,\dots,n$, if $x_i\in V^*$, then $M(X_i)=\top$, otherwise $M(X_i)=\bot$. \qed
\end{proof}

Surprisingly, in a more realistic scenario, where the length of  valley-free paths is bounded by a constant $k$, we have that in the \simpleHijack setting an attack strategy can be found in polynomial time ($n^{O(k)}$, where $n$ is the number of vertices of $G$). Let $N$ be the set of neighbors of $m$. Indeed, the difficulty of the \hijack problem in the \simpleHijack setting depends on the fact that $m$ has to decide to which of the vertices in $N$ it announces the attacked prefix $\pi$, which leads to an exponential number of possibilities. 
However, when the longest valley-free path in the graph is bounded by a constant $k$, it is possible to design a polynomial-time algorithm based on the following intuition, that will be formalized below. Suppose $m$ is announcing  $\pi$ to a subset $A \subseteq N$ of its neighbors and path $p=(z\ \dots\ n\ m)$ is available at an arbitrary vertex $z$ of the graph. Let $n_1,n_2$ be two vertices of $N \setminus A$. If $p$ is disrupted (is not disrupted) by better class both when $\pi$ is announced either to $n_1$ or to $n_2$, then $p$ is disrupted (is not disrupted) by better class when $\pi$  is announced to both $n_1$ and $n_2$. This implies that once $m$ has a candidate path $p^*$ for attracting traffic from $s$, it can check independently to which of its neighbors it can announce $\pi$ without disrupting $p^*$ by better class, which guarantees that a path from $m$ to $z$ longer than $p$ cannot be selected at $z$.

In order to prove Theorem~\ref{theo:grexpall_plb_origin_bgp}, we introduce the following lemmata that relate attacks to the structure of the Internet.


\newcommand{\lemmaDisruptedBySameClass}{
 Consider a valley-free path $p=(v_n\ \dots\ v_1)$ and consider an attack of $m$ such that $v_1$ announces a path $p_{v_1}$  to $v_{2}$ to reach prefix $\pi$ and $p$ is possibly disrupted only by same class. Vertex $v_n$ selects a path $p_n \le_{\lambda}^{v_n}pp_{v_1}$.}

\begin{lemma}\label{lemma:grexpall_origin_same_class_disruption}
\lemmaDisruptedBySameClass
\end{lemma}

\begin{proof}
We prove inductively that each vertex $v_i$ in $p$ selects a path $p_i\le_{\lambda}^{v_i}(v_i\ \dots\ v_1)$ such that $|p_i|\le|(v_i\ \dots\ v_1)|$. In the base case ($n=1$), the statement holds since $p_1\le_{\lambda}^{v_1} p_1$ and $|p_1|\le|p_1|$. In the inductive step ($n>1$), by induction hypothesis and NE policy, vertex $v_i$ receives a path $p_{i-1}$ from vertex $v_{i-1}$ such that $p_{i-1}\le_{\lambda}^{v_{i-1}}(v_{i-1}\ \dots\ v_1)$ and $|p_{i-1}|\le|(v_{i-1}\ \dots\ v_1)|$. Two cases are possible: $p_{i-1}$ contains $v_i$, or not. In the second case, $v_i$ selects a path $p_{i}\le_{\lambda}^{v_i} (v_i\ v_{i-1})p_{i-1}$ and since path $(v_i\ \dots\ v_1)$ is disrupted only by same class, we have also $|p_i|\le|(v_i\ v_{i-1})p_{i-1}|\le|(v_i\ \dots\ v_1)|$. In the first case, let $p'$ be the subpath of $p_{i-1}$ from $v_i$. Observe that since $(v_i\ v_{i-1})p_{i-1}$ is a valley-free path and vertex $v_i$ is repeated in that path, we have that $f^{v_i}(p')>f^{v_i}((v_i\ v_{i-1})p_{i-1})=f^{v_i}(v_i\ \dots\ v_1)$, which is not possible  since $(v_i\ \dots\ v_1)$ cannot be disrupted by better class. 
\qed
\end{proof}

\newcommand{\lemmaDisruptedByBetterClass}{
Consider a successful attack for $m$ and let $p_{sm}$ be the path selected at $s$. 
Let $p_{sd}$ be a valley-free path from $s$ to $d$ such that it does not traverse $m$ and such that $p_{sd} <_\lambda^s p_{sm}$. Path $p_{sd}$ is disrupted by a path of better class.}

\begin{lemma}\label{lemma:grexpall_plain_better_class_disruption}
\lemmaDisruptedByBetterClass
\end{lemma}

\begin{proof}
Suppose by contradiction that there exists  a valley-free path $p_{sd}$ from $s$ to $d$ such that $p_{sd} <_\lambda^s p_{sm}$ and $p_{sd}$ is not disrupted by a path of better class.
If $p_{sd}$ is not disrupted, then it is available at vertex $s$. It implies that $s$ selects $p_{sd}$ as its best path, which leads to a contradiction. Otherwise, suppose $p_{sd}$ is disrupted only by same class. By Lemma~\ref{lemma:grexpall_origin_same_class_disruption} we have a contradiction since $s$ selects a path $p\le_{\lambda}^{s}(p_{sd})<_{\lambda}^{s}(p_{sm})$ different from $p_{sm}$.
\qed
\end{proof}

\begin{lemma}\label{lemma:grexpall_plain_better_class_propagation}
 Let  $p=(v_n\ \dots\ v_1)$ be a valley-free path. Consider an attack where $v_1$ announces a path $p_1$ to $v_{2}$. Vertex $v_n$ selects a path of class  at least $f^{v_n}(p)$.
\end{lemma}

\begin{proof}
We prove that each vertex $v_i$ in $p$ selects a path $p_i$ such that $f^{v_i}(p_i)\ge f^{v_i}(v_i\ \dots\ v_1)$. In the base case ($n=1$), the statement holds since $f^{v_1}(p_1)\ge f^{v_1}(p_1)$. In the inductive step ($n>1$), by induction hypothesis and NE policy, vertex $v_i$ receives a path $p_{i-1}$ from vertex $v_{i-1}$ such that $f^{v_{i-1}}(p_{i-1})\ge f^{v_{i-1}}(v_{i-1}\ \dots\ v_1)$. Two cases are possible: $p_{i-1}$ contains $v_i$ or not. In the second case, $v_i$ selects a path $p_{i}\le_{\lambda}^{v_i} (v_i\ v_{i-1})p_{i-1}$ which implies that $f^{v_i}(p_i)\ge f^{v_i}(v_i\ \dots\ v_1)$. In the first case, let $p'$ be the subpath of $p_{i-1}$ from $v_i$. Observe that since $(v_i\ v_{i-1})p_{i-1}$ is a valley-free path and vertex $v_i$ is repeated in that path, we have that, $f^{v_i}(p')>f^{v_i}((v_i\ v_{i-1})p_{i-1})=f^{v_i}(v_i\ \dots\ v_1)$, and the statement holds also in this case.\qed
\end{proof}


\newcommand{\greBgpOriginPlb}{
If the manipulator has \simpleHijack cheating capabilities and the length of
the longest valley-free path is bounded by a constant, then problem \hijack is
in P.
}

\begin{theorem}\label{theo:grexpall_plb_origin_bgp}
\greBgpOriginPlb
\end{theorem}

\begin{algorithm}[bt]
\caption{Algorithm for the \hijack problem where $m$ has \simpleHijack capabilities and the longest valley-free path in the graph is bounded.}
\label{algo:bgp_origin_plb_grexpall}
\algsetup{indent=2em}
\begin{algorithmic}[1]
\STATE {\bf Input}: instance of \hijack problem, $m$ has \simpleHijack cheating capabilities;
\STATE {\bf Output}: an attack pattern if the attack exists, fail otherwise;
\STATE let $P_{sm}$ be the set of all valley-free paths from $s$ to $m$; 
\FORALL  {$p_{sm}$ in $P_{sm}$}
 \STATE let $w$ be the vertex of $p_{sm}$ adjacent to $m$; let $A$ be a set of vertices and initialize $A$ to $\{w\}$; let $N$ be the set of the neighbors of $m$;
 \FORALL {$n$ in $N \setminus \{ w \}$}
  \IF {there is no path $p$ through $(m,n)$ to a vertex $x$ of $p_{sm}$ such that $f^x(p)>f^x(p_{xm})$, where $p_{xm}$ is the subpath of $p_{sm}$ from $x$ to $m$}
   \STATE insert $n$ into $A$ 
  \ENDIF
 \ENDFOR
 \IF {the attack succeeds when $m$ announces $\pi$ only to the vertices in $A$}
  \RETURN $A$
 \ENDIF
\ENDFOR
\RETURN fail
\end{algorithmic}
\end{algorithm}

\begin{proof}
We tackle the problem with Alg.~\ref{algo:bgp_origin_plb_grexpall}. 
First, observe that line 9 tests if a certain set of announcements causes a successful attack and, in that case, it returns the corresponding set of neighbors to whom $m$ announces prefix $\pi$. Hence, if Alg.~\ref{algo:bgp_origin_plb_grexpall} returns without failure it is trivial to see that it found a successful attack. 
Suppose now that there exists a successful attack $a^*$ from $m$ that is not found by Alg.~\ref{algo:bgp_origin_plb_grexpall}. Let $p_{sm}^*$ be the path selected by $s$ in attack $a^*$. 
Let $ A^*$ be the set of neighbors of $m$ that receives prefix $\pi$ from $m$ in the successful attack.

Consider the iteration of the Alg.~\ref{algo:bgp_origin_plb_grexpall} where path $p_{sm}^*$ is analyzed in the outer loop. At the end of the iteration Alg.~\ref{algo:bgp_origin_plb_grexpall} constructs a set $A$ of neighbors of $m$. Let $a$ be an attack from $m$ where $m$ announces $\pi$ only to the vertices in $A$.

First, we prove that $A^* \subseteq A$.
Suppose by contradiction that there exists a vertex $n\in A^*$ that is not contained in $A$. It implies that there exists a valley-free path $p$ through $(m,n)$ to a vertex $x$ of $p_{sm}^*$ such that $f^x(p)>f^x(p_{xm})$, where $p_{xm}$ is the subpath of $p_{sm}^*$ from $x$ to $m$. Since $m$ announces $\pi$ to $n$, by Lemma~\ref{lemma:grexpall_plain_better_class_propagation}, we have that $x$ selects a path $p'$ of class at least $f^x(p)$, that is a contradiction since $p_{sm}^*$ would be disrupted by better class. Hence, $A^* \subseteq A$.

Now, we prove that attack $a$ is a successful attack for $m$.
Consider a valley-free path $p_{sd}$ from $s$ to $d$ that does not traverse $m$ and is preferred over $p_{sm}^*$. By Lemma~\ref{lemma:grexpall_plain_better_class_disruption} it is disrupted by better class in attack $a^*$. By Lemma~\ref{lemma:grexpall_plain_better_class_propagation}, since $A^*\subseteq A$, we have that also in $a$ path $p_{sd}$ is disrupted by better class. Let $x$ be the vertex adjacent to $s$ in $p_{sd}$. Observe that, vertex $s$ cannot have an available path $(s\ x)p$ to $d$ such that $(s\ x)p <_{\lambda}^{s} p_{sm}^*$, because $(s\ x)p$ must be disrupted by better class. 

Moreover, consider path $p_{sm}^*$. Since in $a^*$ path $p_{sm}^*$ is not disrupted by better class by a path to $d$, by Lemma~\ref{lemma:grexpall_plain_better_class_propagation}, there does not exist a path $p_{xd}'$ from a vertex $x$ of $p_{sm}^*$ to $d$ of class higher than $p_{xm}$, where $p_{xm}$ is the subpath of $p_{sm}^*$ from $x$ to $m$. Hence, path $p_{sm}^*$ cannot be disrupted by better class by a path to $d$. Also, observe that for each $n \in A$ there is no path $p$ through $(m,n)$ to a vertex $x$ of $p_{sm}^*$ such that $f^x(p)>f^x(p_{xm})$, where $p_{xm}$ is the subpath of $p_{sm}$ from $x$ to $m$. Hence, $p_{sm}^*$ can be disrupted only by same class. By Lemma~\ref{lemma:grexpall_origin_same_class_disruption}, we have that $s$ selects a path $p$ such that $p\le_{\lambda}^sp_{sm}^*$. Since path $p$ cannot be a path to $d$, attack $a$ is successful. This is a contradiction since we assumed that Alg.~\ref{algo:bgp_origin_plb_grexpall} failed.

Finally, since the length of the valley-free paths is bounded, the iterations of the algorithm where paths in $P_{sm}$ are considered require a number of steps that is polynomial in the number of vertices of the graph. Also, the disruption checks can be performed in polynomial time by using the algorithm in~\cite{ssz-ssir-09}.
\qed
\end{proof}

%% file: sbgp.tex
\section{\SBGP Gives Hackers Hard Times}\label{sect:sbgp}

We open this section by strengthening the role of \SBGP as a security protocol. Indeed, \SBGP adds more complexity to the problem of finding an attack strategy (Theorem~\ref{theo:grexpall_plb_sbgp}). 
After that we also provide an answer to a conjecture posed in~\cite{gshr-hsasirp-10} about hijacking and interception attacks in \SBGP when a single path is announced by the manipulator. In this case, we prove that every successful hijacking attack is also an interception attack (Theorem~\ref{theo:gr_sbgp_disruption}).



\begin{theorem}\label{theo:grexpall_plb_sbgp}
If the manipulator has \SBGP cheating capabilities and the length of the longest valley-free path is bounded by a constant, then problem \hijack is \NPhard.
\end{theorem}

\begin{proof}
We reduce from a version of  \threesat where each variable appears at most three times and each positive literal at most once~\cite{p-cc-94}. Let $F$ be a logical formula in conjunctive normal form  with variables $X_1 \dots X_n$ and clauses $C_1 \dots C_h$. We build a
BGP instance $G$ (see Fig.~\ref{fig:bgp_sbgp_plb_grexpall}) consisting of $4$ structures: \textsc{Intermediate}, \textsc{Short}, \textsc{Long}, and \textsc{Disruptive}.
%
%
The \longLengthStruct is a directed path of length $6$ with edges $(s,w_1)$, $(w_1,w_2)$, \dots, $(w_{4},w_{5})$, and~$(w_5,d)$. 
The \intLengthStruct consists of a valley-free path joining $m$ and $s$. It has length $4$ and it is composed by a directed path $(s\ j_3\ j_2\ j_1)$, and a directed edge~$(m,j_1)$.
The \shoLengthStruct has $h$ directed paths from $s$ to $d$. Each path has length at most $4$ and has edges $(s,c_{i,1})$, $(c_{i,1},c_{i,2})$, \dots, $(c_{i,v(C_i)},d)$ ($1\leq i\leq h$), where $v(C_i)$ is the size of  $C_i$.
The \disruptiveLengthStruct contains, for each variable $X_i$  vertices, $r_i$, $t_i$, $x_i$, $p_i$ and $p'_i$. 
Vertices, $r_i$, $t_i$, and $x_i$, are reached via long directed paths from $m$ and  are connected by  $(t_i, p_i)$, $(x_i, p_i)$, $(x_i, p'_i)$, $(r_i, j_3)$, $(p_i, j_3)$, and $(p_i, d)$.
Finally, suppose $X_i$ occurs in clause $C_j$ with a literal in position $l$. If the literal is negative the undirected edge $(p_i,c_{j,l})$ is added, otherwise, edges $(p_i,c_{j,l})$, $(r_i,c_{j,l})$, $(c_{j,l},j_3)$, and undirected edge $(p'_i,c_{j,l})$ are added. 
An edge connects $m$ to $d$.
Vertices $s$, $d$, and $m$ have source, destination, and manipulator roles, respectively. 

Intuitively, the proof works as follows. The paths that allow traffic to go from $s$ to $m$ are only those passing through the \disruptiveLengthStruct and the one in the \intLengthStruct. Also, the path through the \intLengthStruct is shorter than the one through the \longLengthStruct, which is shorter than those through the \disruptiveLengthStruct.
If $m$ does not behave maliciously, $s$ receives only paths traversing the \shoLengthStruct and the \longLengthStruct. In this case $s$ selects one of the paths in the \shoLengthStruct according to its tie break policy.
If $m$ wants to attract traffic from $s$, then:
\begin{inparaenum}[(i)]
 \item  path $(j_3\  j_2\  j_1\ m\ d)$ must be available at $s$ and
 \item all paths contained in the \shoLengthStruct must be disrupted by a path announced by $m$.
\end{inparaenum}
If (i) does not hold, then $s$ selects the path contained in either the \longLengthStruct or  the \shoLengthStruct. If (ii) does not hold, then $s$ selects a path contained in the \shoLengthStruct.

Our construction is such that the \threesat formula is true iff $m$ can attract the traffic from $s$ to $d$.  To understand the relationship with the \threesat problem, consider the behavior of $m$ with respect to variable $X_1$ (see Fig.~\ref{fig:bgp_sbgp_plb_grexpall}) that appears with a positive literal in the first position of clause $C_1$, a negative literal in the first position of $C_2$ and a negative literal in the second position of $C_h$. 


First, we explore the possible actions that $m$ can perform in order to disrupt paths in the \shoLengthStruct. Since $m$ has \SBGP cheating capabilities, $m$ is constrained to propagate only the announcements it receives. If $m$ does not behave maliciously, $m$ receives path $(d)$ from $d$ and paths $P_{r_1}$, $P_{t_1}$, and $P_{x_1}$ from $r_1$, $t_1$, and $x_1$, respectively. These paths have the following properties: $P_{r_1}$ contains vertex $c_{1,1}$ that is contained in the path of the \shoLengthStruct that corresponds to clause $C_1$; paths $P_{t_1}$ and $P_{x_1}$ both contain vertex $p_1$ and do not contain vertex $c_{1,1}$ since $p_1$ prefers $(p_1\ d)$ over $(p_1\ c_{1,1}\ c_{1,2}\ c_{1,3}\ d)$.

Now, we analyze what actions are not useful for $m$ to perform an attack. If $m$ issues any announcement towards $t_1$ or $r_1$ the path traversing the \intLengthStruct is disrupted by better class. Also, if $m$ sends a path $P_{r_1}, P_{t_1}$, or $P_{x_1}$ towards  $r_j$, $t_j$, or $x_j$, with $j=2,\dots,n$, the path traversing the \intLengthStruct is disrupted by better class.
Also, if $m$ sends $(m\ d)$ to $x_1$, then the path traversing the \intLengthStruct is disrupted from $c_{1,1}$ by better class. If $m$ sends $P_{x_1}$ to $x_1$, then it is discarded by $x_1$ because of loop detection. In each of these cases $m$ cannot disrupt any path traversing the \shoLengthStruct without disrupting the path traversing the \intLengthStruct. 
Hence, $m$ can disrupt path in the \shoLengthStruct without disrupting the path traversing the \intLengthStruct  only announcing $P_{r_1}$ and $P_{t_1}$ from $m$ towards $x_1$. 


If path $P_{t_1}$ is announced to $x_1$, then $p_1$ discards that announcement because of loop detection and path $(s\ c_{1,1}\ c_{1,2}\ c_{1,3}\ d)$  is disrupted from $p'_1$ by better class. Also, the path through the \intLengthStruct remains available because the announcement through $p'_1$ cannot reach $j_3$ from $c_{1,1}$, otherwise valley-freeness would be violated. Hence, announcing path $P_{t_1}$, corresponds to assigning true value to variable $X_1$, since the only path in the \shoLengthStruct that is disrupted is the one that corresponds to the clause that contains the positive literal of $X_1$.

If path $P_{r_1}$ is announced to $x_1$, then $c_{1,1}$ discards that announcement because of loop detection and both paths $(s\ c_{2,1}\ c_{2,2}\ c_{2,3}\ d)$ and $(s\ c_{h,1}\ c_{h,2}\ c_{h,3}\ d)$ are disrupted by better class from $p_1$. Also, the path through the \intLengthStruct remains available because the announcement through $p_1$ cannot reach $j_3$ from $c_{2,1}$ or $c_{h,2}$, otherwise valley-freeness would be violated. Hence, announcing path $P_{r_1}$, corresponds to assigning false value to variable $X_1$, since the only paths in the \shoLengthStruct that are disrupted are the ones that correspond to the clauses that contain a negative literal of $X_1$.

Hence, announcing path $P_{t_1}$ ($P_{r_1}$) from $m$ to $x_1$ corresponds to assigning the true (false) value to variable $X_1$. As a consequence, $m$ can disrupt every path in the \shoLengthStruct without disrupting the path in the \intLengthStruct iff formula $F$ is satisfiable.\qed
\end{proof}

\newcommand{\greSbgpNdb}{
If the manipulator has \SBGP cheating capabilities and its degree is bounded by a constant, then problem \hijack is in P.}

\begin{theorem}\label{theo:grexpall_plb_ndb_sbgp}
\greSbgpNdb
\end{theorem}
\begin{proof}
Observe that if the manipulator $m$ has \SBGP cheating capabilities, the degree of
the manipulator's vertex is bounded by a constant $k$, then problem \hijack is
in P. In fact, since $m$ has at most $k$ available paths plus the empty path, a brute force approach approach needs to explore $(k+1)^k$ number of possible cases.\qed
\end{proof}

To study the relationship between hijacking and interception we introduce the following technical lemma.

\newcommand{\lemmaSbgpDisruption}{
Let $G$ be a \GRexpall compliant BGP instance, let $m$ be a vertex with \SBGP cheating capabilities, and let $d \neq m$ be any vertex of $G$. All vertices that admit a class $c$ valley-free path to $d$ not containing $m$ have an available path of class $c$ or better to $d$, irrespective of the paths propagated by $m$ to its neighbors.}
\begin{lemma}\label{lemm:gr_sbgp_disruption}
\lemmaSbgpDisruption
\end{lemma}
\begin{proof}
Let $p=(v_n\ \dots\ v_1)$ be a valley-free path to $d$ not containing $m$. 
We prove by induction on vertices $v_1,\dots,v_n$ that each vertex $v_i$ has an available path of class $f^{v_i}(v_i\ \dots\ v_1)$ or better. In the base case $i=2$,  $v_2$ is directly connected to $d$ and the statements trivially holds.
Suppose that vertex $v_{i}$, with $i > 2$, has an available path of class $f^{v_{i}}(v_{i}\ \dots\ v_1)$.
Hence, $v_i$ selects a path $p^*$ such that $f^{v_{i}}(p^*)\ge f^{v_{i}}(v_{i}\ \dots\ v_1)$. Also, since 
$(v_{i+1}\ v_{i}\ \dots\ v_1)$ is valley-free even $(v_{i+1}\ v_i)p^*$ is valley-free. Then, $v_{i}$ announces (because of the NE policy) its best path $p^*$ to $v_{i+1}$. 
There are two possible cases: either $p^*$ does not contain $v_{i+1}$ or not. In the first case, path $(v_{i+1}\ v_i)p^*$ is available at  $v_{i+1}$ and the statement holds. In the second case, consider the subpath $p^*_{v_{i+1}}$ of $p^*$ from $v_{i+1}$ to $d$. The statement easily follows because $f^{v_{i+1}}(p^*_{v_{i+1}})\ge f^{v_{i+1}}((v_{i+1}\ v_i)p^*$. \qed
\end{proof}

\begin{theorem}\label{theo:gr_sbgp_disruption}
Let $m$ be a manipulator with \SBGP cheating capabilities.
If $m$ announces the same path to any arbitrary set of its neighbors, then every successful hijacking attack is also a successful interception attack.
If $m$ announces different paths to different vertices, then the hijacking may not be an interception.
\end{theorem}

\begin{proof}
We prove the following more technical statement that implies the first part of the theorem. 
Let $G$ be a BGP instance, let $m$ be a vertex with \SBGP cheating capabilities. Let $p$ be a path available at $m$ in the stable state $S$ reached when $m$ behaves correctly. Suppose that $m$ starts announcing $p$ to any subset of its neighbors. Let $S'$ be the corresponding routing state. Path $p$ remains available at vertex $m$ in $S'$.
The truth of the statement implies that $m$ can forward the traffic to $d$ by exploiting $p$.

Suppose for a contradiction that path $p$ is disrupted in $S'$ when $m$ propagates it to a subset of its neighbors. Let $x$ be the first vertex of $p$ that prefers a different path $p_x$ ($p$ is disrupted by $p_x$) in $S'$ and let $p'$ be the subpath of $p$ from vertex $d$ to $x$ (see Fig.~\ref{fig:grexpall_sbgp_no_disruption}).
Observe that $p$ is not a subpath of $p_x$ as $x$ cannot select a path that passes through itself. 
Since $p_x$ is not available at $x$ in $S$, let $y$ be the vertex in $p_x$ closest to $d$ that selects a path $p_y$ that is preferred over $p_x'$ in $S$, where $p_x'$ is the subpath of $p_x$ from  $y$ to $d$.

\begin{figure}[t]
\begin{minipage}[t]{0.59\textwidth}%
\vspace{-30mm} 
\includegraphics[width=1\textwidth]{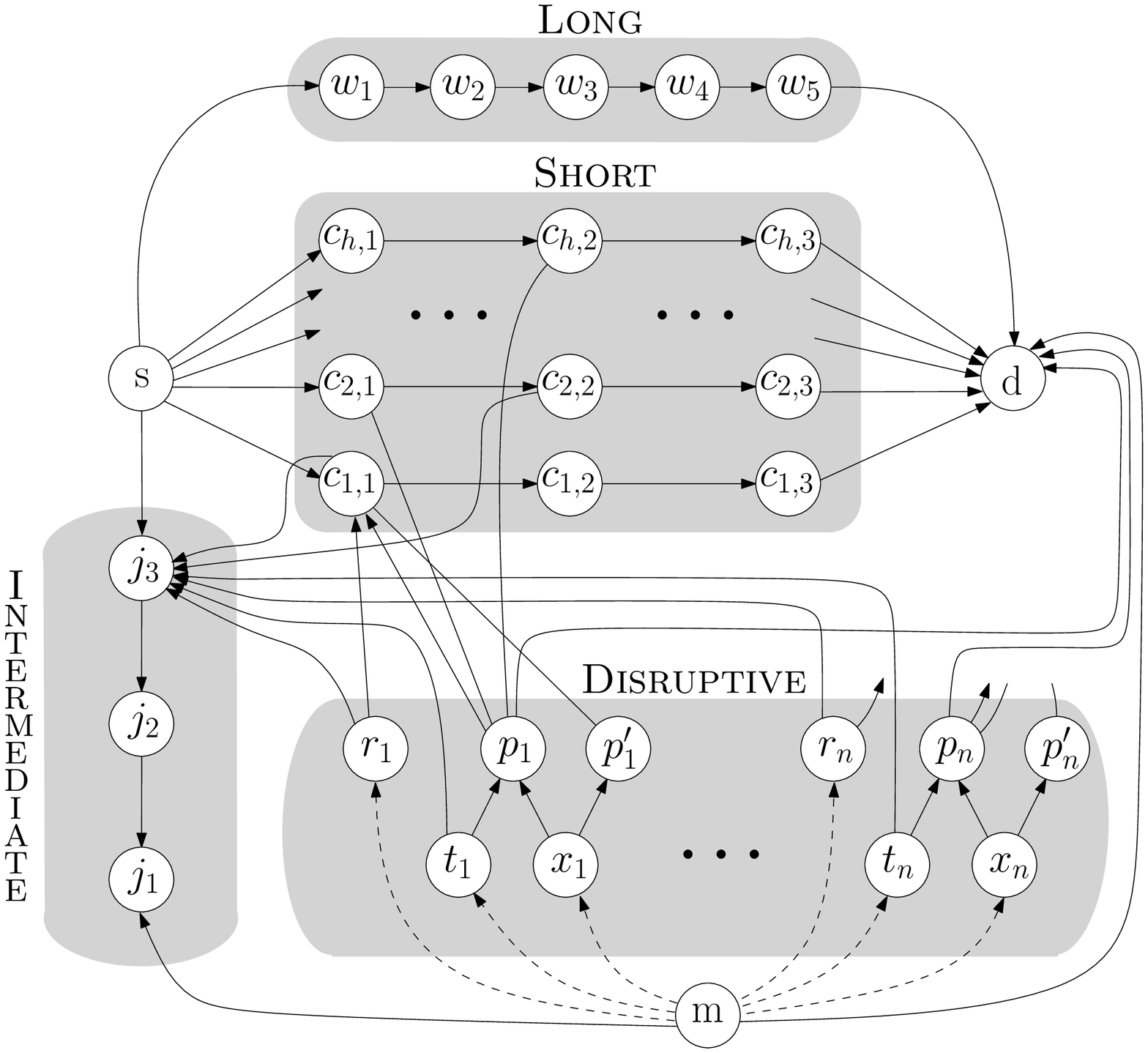}
\caption{Reduction of a constrained \threesat problem to the \hijack problem when $m$ has \SBGP cheating capabilities.}\label{fig:bgp_sbgp_plb_grexpall}
\end{minipage}
\hfill
\begin{minipage}[t]{0.39\textwidth}%
\centering
\includegraphics[width=0.8\textwidth]{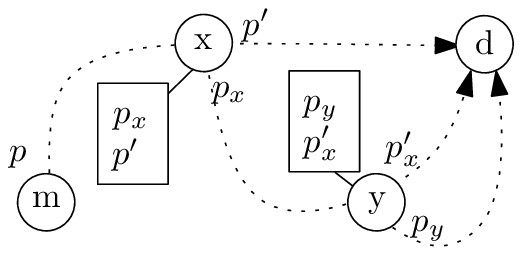}
\caption{Proof of Theorem~\ref{theo:gr_sbgp_disruption}. (a) The order of paths into the boxes represents the preference of the vertices.}\label{fig:grexpall_sbgp_no_disruption}
\includegraphics[width=0.5\textwidth]{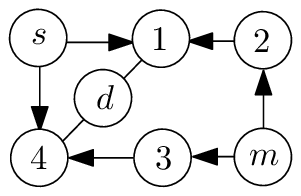}
\caption{An instance where $m$ cannot intercept traffic to $d$ but it can hijack it.}\label{fig:counterexample_gr_sbgp_disruption_yes}
\end{minipage}%

\end{figure}

We have two cases: either $f^x(p_x)>f^x(p')$ or $f^x(p_x)=f^x(p')$ (i.e., $p_x$ is preferred to $p'$ by better or by same class).

Suppose that $f^x(p_x)>f^x(p')$. By Lemma~\ref{lemm:gr_sbgp_disruption}, since there exists a valley-free path $p_x$ from $x$ to $d$ that does not traverse $m$, then $x$ has an available path of class at least $f^x(p_x)$. Hence, $x$ cannot select path $p'$ in $S$,  a contradiction. 

Suppose that $f^x(p_x)=f^x(p')$. Two cases are possible: either $p_y$ contains $x$ or not. In the first case either $f^y(p_y) > f^y(p_x')$ or $f^y(p_y) = f^y(p_x')$. If $f^y(p_y) > f^y(p_x')$, then we have that $f^{y}(p_{y}) \le f^{x}(p') = f^{x}(p_x) \le f^{y}(p_x')$, a contradiction. If $f^y(p_y) = f^y(p_x')$, we have that $|p_x'|<|p_x|\le|p'|<|p_y|$. A contradiction since a longer path is preferred.

The second case ($f^x(p_x)=f^x(p')$ and $p_y$ does not contain $x$) is more complex. We have that $|p'|\ge|p_x|$. Also, by Lemma~\ref{lemm:gr_sbgp_disruption}, since $p_y$ and $p_x'$ do not pass through $m$, then $y$ has an available path of class at least $\max\{f^y(p_y),f^y(p_x')\}$. As $y$ alternatively chooses $p_y$ and $p_x'$ we have that $f^y(p_y)=f^y(p_x')$, which implies that $|p_x'|\ge|p_y|$. Denote by $p_{xy}$ the subpath $(v_m\ \dots v_0)$ of $p_x$, where $v_0=y$ and $v_m=x$. Consider routing in state $S$. Two cases are possible: either $p_{xy}p_y$ is available at $x$ or not. In the first case, since $|p'|\ge|p_x|=|p_{xy}p_x'|\ge|p_{xy}p_y|$, we have a contradiction because $p'$ would not be selected in $S$. In the second case, we will prove that for each vertex $v_h \neq x$ in $p_{xy}$ we have that $|p_h|\le|(v_h\ \dots\ v_0)p_y|$, where $p_h$ is the path selected by $v_h$ in $S$. This implies that $|(v_m\ v_{m-1})p_{m-1}|\le|p_{xy}p_y|\le|p_x|\le|p'|$ and this leads to a contradiction. In fact, if $|(v_m\ v_{m-1})p_{m-1}|<|p'|$, then we have a contradiction because $p'$ would not be selected in $S$. Otherwise, if $|(v_m\ v_{m-1})p_{m-1}|=|p'|$, we have that $|p_x|=|p'|$. Then, $x$ prefers $p_x$ over $p'$ because of tie break. We have a contradiction since also $(v_m\ v_{m-1})p_{m-1}$ is preferred over $p'$ because of tie break in $S$. 

Finally, we prove that for each vertex $v_h \neq x$ in $p_{xy}$ we have that $|p_h|\le|(v_h\ \dots v_0)p_y|$. This trivially holds for $v_0=y$. We prove that if it holds for $v_i$ then it also holds for $v_{i+1}$. If $v_{i+1}$ selects $(v_{i+1}\ v_i)p_i$, then the property holds. Otherwise, $(v_{i+1}\ v_i)p_i$ is disrupted either by better class or by same class by a path $p_{i+1}$. In the first case, we have that either $p_{i+1}$ traverses $m$ or not. 
Suppose $p_{i+1}$ traverses $m$ and 
let $q'$ be the neighbor of $v_{i+1}$ on $p_{i+1}$.
Since $p_{i+1}$ disrupts $(v_{i+1}\ v_i)p_i$ by better class, then $p_{i+1}$ is composed by a directed path from $d$ to $q'$ and an edge $(q',v_{i+1})$ that can be either an oriented edge from $q'$ to $v_{i+1}$ or an unoriented edge. Let $n$ be the neighbor of $m$ on $p$ and $n'$ be the neighbor of $n$ on $p$ different from $m$. Consider the relationship between $n$ and $n'$. Suppose $n$ is a customer or a peer of $n'$. If $m$ is a provider or a peer of $n$, then $p$ is not valley-free and $p$ cannot be available at $m$ in $S$, which leads to a contradiction. Otherwise, if $m$ is a customer of $n$, then $n$ would have preferred the best path from its customer $m$ rather than the path learnt from its provider $n'$. It implies that $p$ would not be available at $m$ in $S$, that is a contradiction. Hence, $n$ is a provider of $n'$ and the subpath of $p$ from $d$ to $n$ is a directed path. Since  $f^x(p_x)=f^x(p')$, we have that also $p_x$ is a directed path from $d$ to $x$. Therefore, $v_{i+1}$ is a provider of $v_i$ and so $(v_{i+1}\ v_i)p_i$ would not be disrupted by better class in $S$, which is a contradiction.
Hence, $p_{i+1}$ does not  traverse $m$. By Lemma~\ref{lemm:gr_sbgp_disruption}, a path of a class better than $(v_{i+1}\ \dots\ v_0)p_x'$ is available at $v_{i+1}$ and so $v_{i+1}$ cannot select  $(v_{i+1}\ \dots\ v_0)p_x'$ in $S'$, a contradiction. In the second case ($(v_{i+1}\ v_i)p_i$ is disrupted  by same class by a path $p_{i+1}$) we have that $|p_{i+1}|\le|(v_{i+1}\ v_i)p_{i}|\le|(v_{i+1}\ \dots\ v_0)p_y|$. The second inequality comes from the induction hypothesis.

This concludes the first part of the proof. For proving the second part we show an example where $m$ announces different paths to different neighbors and the resulting hijacking is not an interception. Consider the BGP instance in Fig.~\ref{fig:counterexample_gr_sbgp_disruption_yes}. 
In order to hijack traffic from $s$, vertices $1$ and $4$ must be hijacked. Hence, $m$ must announce $(m\ 3\ 4\ d)$ to $2$ and $(m\ 2\ 1\ d)$ to $3$. However, since $(3\ 4\ d)$ and $(2\ 1\ d)$ are no longer available at $m$ the interception fails.
\qed
\end{proof}



\remove{

\newcommand{\grSbgpTwoPaths}{
Suppose that in the stable state of a \GRexpall compliant BGP instance $G$, a manipulator $m$ with \SBGP cheating capabilities announces different paths to different neighbors. Then, one of its announced path may be no more available at vertex $m$.
}

\begin{theorem}\label{theo:gr_sbgp_two_paths}
\grSbgpTwoPaths
\end{theorem}

\begin{proof}
Consider the BGP instance depicted in in Fig.~\ref{fig:counterexample_gr_sbgp_disruption_yes}. If manipulator $m$ announces to its neighbors $2$ and $3$ its available paths $(m\ 3\ 1\
0)$ and $(m\ 2\ 1\ 0)$ to vertices $2$ and $3$, respectively, then both $2$ and $3$ select these path as their best path because they are received from customers. 

\begin{figure}
  \centering
  \subfloat[The BGP instance used in the proof of Theorem~\ref{theo:gr_sbgp_two_paths}.]{\label{fig:counterexample_gr_sbgp_disruption_yes}\includegraphics[width=0.27\columnwidth]{figures/counterexample_gr_sbgp_disruption_yes2}}\ \ \ \ \ \ \ \ \ \                
  \subfloat[The BGP instance used in the proof of Theorem~\ref{theo:grs-sbgp-disruption}.]{\label{fig:counterexample_gr_sbgp_disruption}\includegraphics[width=0.60\columnwidth]{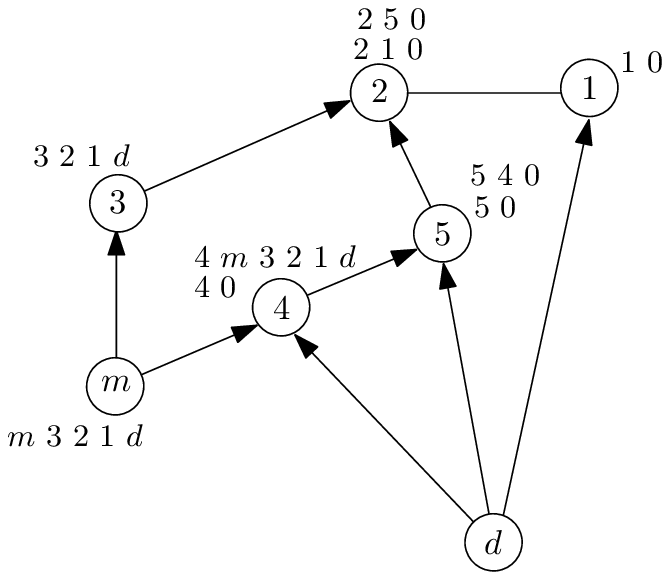}}
  \caption{Examples of attacks.}
  \label{fig:examples_of_attacks}
\end{figure}

\end{proof}

}